\documentclass[12pt]{amsart}
\usepackage{hyperref, amsmath,amssymb,amsfonts,amsbsy,mathtools,cite,setspace}
\usepackage{graphicx}
\usepackage{amssymb}
\usepackage{amsthm}
\usepackage{enumerate}
\usepackage{wasysym}
\usepackage{cite}
\usepackage{tikz}
\usepackage{mathtools}

\newcommand{\Alg}{Alg}

\def\ket#1{| #1 \rangle}

\def\kb#1#2{|#1\rangle\!\langle #2 |}

\newtheorem{defn}{Definition}
\newtheorem{prop}{Proposition}
\newtheorem{thm}{Theorem}
\newtheorem{exa}{Example}
\newtheorem{lem}{Lemma}

\newtheorem{rmk}{Remark}
\newtheorem{cor}{Corollary}

\newcommand{\C}{{\mathbb C}}

\newcommand{\operp}{$\bigcirc$\kern-.91em{$\perp$}}

\DeclarePairedDelimiter\floor{\lfloor}{\rfloor}

\begin{document}
\title[Private Subsystems and Quasiorthogonal Algebras]{Private Quantum Subsystems and Quasiorthogonal Operator Algebras}
\author[J.Levick, T.Jochym-O'Connor, D.W.Kribs, R.Pereira, R.Laflamme]{Jeremy Levick$^{1,2}$, Tomas Jochym-O'Connor$^{3}$, David W. Kribs$^{1,3}$, Raymond Laflamme$^{3,4,5}$, Rajesh Pereira$^{1}$}

\address{$^1$Department of Mathematics \& Statistics, University of Guelph, Guelph, ON, Canada N1G 2W1}
\address{$^2$African Institute for Mathematical Sciences, Muizenberg, Cape Town, 7945, South Africa}
\address{$^3$Institute for Quantum Computing and Department of Physics \& Astronomy, University of Waterloo, Waterloo, ON, Canada N2L 3G1}
\address{$^4$Perimeter Institute for Theoretical Physics, Waterloo, ON, Canada N2L 2Y5}
\address{$^5$Canadian Institute for Advanced Research, Toronto, ON, Canada M5G 1Z8}

\begin{abstract} We generalize a recently discovered example of private quantum subsystem to find private subsystems for Abelian subgroups of the $n$-qubit Pauli group, which exist in the absence of private subspaces. In doing so, we also connect these quantum privacy investigations with the theory of quasiorthogonal operator algebras through the use of tools from group theory and operator theory.
\end{abstract}



\maketitle

\section{Introduction}
A key goal of quantum information theory is to understand how concepts from classical information translate into the quantum setting. For example, in classical information theory, one finds the idea of the \emph{one-time pad}, a system that can be used to encode information sent through a channel, so long as the two parties share a private key between them. The generalization of this scheme to the quantum setting was established in \cite{ambainis,boykin};  initially called private quantum channels and now referred to as \emph{private quantum subsystems}. Developing a characterization of the special features of particular forms of completely positive maps is central to the study.

Subsequent developments on the subject over the past fifteen years have included advances in the theory of private shared reference frames \cite{bartlett1,bartlett2}, applications to quantum secret sharing \cite{cleve,crepeau}, bridges established with quantum error correction \cite{kretschmann}, and a first indication of connections with the theory of operator algebras \cite{church}. More recently, a significant step was made toward a general theory for private quantum subsystems in \cite{jochym,jochym2}, including algebraic conditions that characterize such quantum privacy. In addition, that work included the first example of a private subsystem for a quantum channel that exists in the absence of any subspace that is privatized by the channel. The example in question was a simple two-qubit dephasing channel, and it was initially surprising to discover the underlying structure of a private subsystem for the channel.

This article began with an effort to generalize the example from \cite{jochym,jochym2} more widely; in particular to obtain a large class of simple and easily implemented $n$-qubit quantum channels, which as a consequence of their simplicity would not privatize any subspace, but that nevertheless do privatize quantum information through more delicate subsystem encodings. In addition to building on previous work, one could imagine such a class of channels and encodings as being useful in quantum computing when deleting information swiftly is a desired outcome. As an outgrowth of our example-based effort, we also uncover an advance in the general theory of private quantum subsystems. Specifically, we show how finite-dimensional private quantum subsystems can be cast in the framework of {\it quasiorthogonal operator algebras}, which have been studied in quantum information for their intrinsic interest \cite{weiner,ohnopetz,ohno,petzkahn} and specific connection with identifying maximal collections of mutually unbiased bases \cite{band,pittenger,wocjan}. Combining the quasiorthogonal perspective with conditional expectation tools from operator algebra \cite{pereira} and a group-theoretic analysis \cite{rudin}, we then show how quantum channels defined by operators forming an Abelian subgroup of the Pauli group can privatize quantum information encoded into subsystems of $n$-qubit Hilbert space, even though they can never privatize qubits directly encoded into subspaces of the Hilbert space. This gives a full picture of privatizing quantum information for such subgroups.

This paper organized as follows. In the next section we present preliminary material on private quantum subsystems. This is followed by a discussion on  quasiorthogonal operator algebras and conditional expectations, and we indicate how private subsystems fit into the framework. In the next section we turn our attention to the structure of Abelian subgroups of the Pauli group. We conclude by explicitly showing how such subgroups can be used to privatize quantum information.


\section{Private Quantum Subsystems and Subspaces}

The motivating communication scenario for a private quantum code starts with two parties, Alice and Bob, who share a private classical key that Alice uses to inform Bob which of a set of unitary operators $\{ U_i \}$ she has used to encode her quantum state (density operator) $\rho \mapsto U_i \rho U_i^\dagger$. Bob can then decode and recover the state $\rho$ without disturbing it. The set of unitaries $\{U_i\}$ and the probability distribution $\{p_i\}$ that makes up the random key which determines the encoding unitary are shared publicly. Hence an eavesdropper Eve's description of the system is given by the random unitary channel $\Phi(\rho) = \sum_i p_i U_i \rho U_i^\dagger$. By carefully selecting the unitary operators and probabilities, the random unitary channel will provide Eve with no information about the input state.

More generally, quantum channels are given by completely positive trace-preserving maps $\Phi:M_n(\C)\rightarrow M_n(\C)$, and in the context of private communication one looks for sets $S$ of density operators $\rho$ that are mapped by $\Phi$ to the same state, $\Phi(\rho) = \rho_0$ for all $\rho\in S$. Analogous to quantum error correction, we focus on sets of states that are associated with underlying Hilbert space structure and hence form a true quantum code in the sense that superpositions of privatized states are also private. We define a qubit to be a two-dimensional subspace of~$\C^n$ whose observables are given by the Pauli matrices in that fixed subspace,
\begin{align*}
X = \begin{pmatrix}
0 & 1 \\
1 & 0
\end{pmatrix},  \ 
Y = \begin{pmatrix}
0 & -i \\
i & 0
\end{pmatrix},  \ 
Z = \begin{pmatrix}
1 & 0 \\
0 & -1
\end{pmatrix}.  
\end{align*}
Thus the most basic choice for a set of private states is for $S$ to be the set of density matrices whose non-zero action is entirely on a $k$-qubit, $2^k$-dimensional subspace of the Hilbert space $\C^n$; i.e., there exists a subspace $\mathcal{C}$ such that the operators in $S$, when written in any basis containing $\mathcal{C}$, has the block form
\[ \rho=\begin{pmatrix} \rho_{\mathcal{C}} & 0 \\ 0 & 0 \end{pmatrix}. \]
As a convenience we shall identify the full set of such operators with the set of operators $\mathcal{L}(\mathcal{C})$ that act on $\mathcal{C}$, which in this case is isomorphic to $\mathcal{L}(\C^{2k})$ and thus $\Phi$ privatizes $k$ qubits of information. Two of the simplest examples in the single qubit case are given by the complete depolarizing channel, which satisfies $\Phi(\rho)=\frac12 I$ and $I$ is the identity operator, and the spontaneous emission channel, which satisfies $\Phi(\rho)= \kb{0}{0}$ for all single qubit $\rho$ \cite{nielsen}. The complete depolarizing channel in particular is the random unitary channel implemented by the equally weighted Pauli operators $I,X,Y,Z$.

However, not all quantum channels admit a private subspace. In particular, as proved in \cite{jochym,jochym2}, any quantum channel $\Phi(\rho) = \sum_i A_i \rho A_i^\dagger$ whose (normal) Kraus operators $A_i$ are mutually commuting, cannot privatize a (multi-dimensional) subspace. Nevertheless, even in such instances it can still be possible to privatize quantum information. The following phase flip example was presented and discussed in detail in \cite{jochym,jochym2}: Let $\Phi:\mathcal{L}(\C^2\otimes\C^2)\rightarrow\mathcal{L}(\C^2\otimes\C^2)$ be the completely positive map on two-qubit space whose Kraus operators are $\{\frac{1}{2}II,\frac{1}{2}Z I, \frac{1}{2}I Z,\frac{1}{2} ZZ\}$, where we have suppressed tensor symbols so that $ZI$ is $Z\otimes I$, etc. These operators are all diagonal in the computational basis $\{ \ket{00}, \ket{01}, \ket{10}, \ket{11}   \}$, and form an orthonormal basis for the diagonal subalgebra of $M_4(\C)$ in the normalized trace (Hilbert-Schmidt) inner product. In particular they are mutually commuting and so the channel $\Phi$ has no private subspaces. However, the subalgebra spanned by $II, I X, YY, YZ$ is privatized by $\Phi$, and moreover, is isomorphic to the Pauli algebra on one qubit; thus the channel still privatizes a qubit:
\[
\Phi\Big(\frac{1}{4}( II + c_2 I X + c_3Y Y + c_4 YZ)\Big) = \frac{1}{4}I I.
\]

This example relies on the more subtle notion of quantum subsystem to privatize a qubit. A space $A$ (or $B$) is a {\it subsystem} of a Hilbert space $\mathcal{H}$ if there is a subspace $\mathcal{C}$ of $\mathcal{H}$ such that $\mathcal{C}$ admits a tensor decomposition as $\mathcal{C} = A \otimes B$.

\begin{defn}
Let $\Phi:\mathcal{L}(\mathcal{H})\rightarrow\mathcal{L}(\mathcal{H})$ be a channel and let $B$ be a subsystem of $\mathcal{H}$. Then $B$ is a \emph{private subsystem} for $\Phi$ if there is a $\rho_0 \in \mathcal{L}(\mathcal{H})$ and $\sigma_A \in \mathcal{L}(A)$ such that
\begin{eqnarray}\label{private_definition}
\Phi (\sigma_A \otimes \sigma_B) = \rho_0 \quad \forall \sigma_B \in \mathcal{L}(B).
\end{eqnarray}
\end{defn}

As shown in \cite{jochym,jochym2}, the above example channel privatizes a qubit subsystem in the sense of Eq.~(\ref{private_definition}), yet does not privatize a subspace. Moreover, the mapping is ``non-operator'' in the language of \cite{jochym,jochym2} in the sense that Eq.~(\ref{private_definition}) holds for a particular state~$\sigma_A \in \mathcal{L}(A)$ but not for every state on $A$. We seek to understand this example and generalize it, and we shall make use of conditional expectations and quasiorthogonality of two subalgebras, notions to which we now turn.

\section{Quasiorthogonal Subalgebras and Conditional Expectations}

We begin this section by reviewing basic notions from operator algebras and then we indicate how private subsystems can be cast in the quasiorthogonal subalgebra framework and discuss examples.

A {\it unital $*$-subalgebra} of the $N\times N$ complex matrices $M_N(\C)$ is a subset $\mathcal{A}\subseteq{M_N(\C)}$ that is closed under matrix addition, scalar multiplication, multiplication, and taking adjoints, and includes the identity matrix $I_N$, which we denote by $I$ when the dimension is clear. It follows from the representation theory of finite-dimensional C$^*$-algebras \cite{davidson}, that any unital $*$-subalgebra $\mathcal{A}$ of $M_N(\C)$ is unitarily equivalent to an algebra of the form
\begin{eqnarray}\label{algunitequiv}
\bigoplus_{i=1}^m I_{k_i}\otimes M_{q_i}(\C),
\end{eqnarray}
where $\sum_i k_iq_i = N$.


\begin{defn} Let $\mathcal{A}$ be a unital $*$-subalgebra of $M_N(\C)$. Then $\Phi_{\mathcal{A}}:M_N(\C)\rightarrow M_N(\C)$ is called the \emph{trace-preserving conditional expectation onto $\mathcal{A}$} when it is the unique completely positive trace-preserving map satisfying
\begin{enumerate}
\item $\Phi_{\mathcal{A}}(a)=a \quad \forall a\in\mathcal{A}$
\item $\Phi_{\mathcal{A}}(a_1xa_2)=a_1\Phi_{\mathcal{A}}(x)a_2 \quad \forall a_1,a_2\in\mathcal{A}, \quad \forall x\in M_N(\C)$
\item $\Phi_{\mathcal{A}}(x)\geq 0 \quad \forall x\geq 0$.
\end{enumerate}
\end{defn}

The trace-preserving conditional expectation onto a unital $*$-sub-algebra is unique as a linear map, and is always a completely positive map.
Any unital $*$-subalgebra induces an inner product, the \emph{left-regular trace inner product}, which is the Hilbert-Schmidt inner product in the induced left-regular trace given by the subalgebra $\mathcal{A}$. More explicitly, let $a\in\mathcal{A}$ and consider the left-regular representation of $a$, $L_a(x)=ax$ for $x\in M_N(\C)$. The left-regular trace of $a$ is the trace of the operator $L_a$; $tr_{\mathcal{A}}(a) := tr(L_a)$. This gives the left-regular trace inner product as $\langle a,b\rangle_{\mathcal{A}} = tr_{\mathcal{A}}(a^*b)$.

Denote by $\mathcal{A}^{\prime}$ the \emph{commutant} of $\mathcal{A}$; $\mathcal{A}^{\prime}=\{x\in M_N(\C): xa = ax \quad \forall a\in\mathcal{A} \}$. Importantly, the Kraus operators for $\Phi_{\mathcal{A}}$ must be an orthonormal basis for $\mathcal{A}^{\prime}$ in the left-regular inner product induced by $\mathcal{A}$. In addition, recall that two different choices of Kraus operators must be related by a partial isometry conjugation. For more on conditional expectations onto matrix algebras, including for proofs and further references, see \cite{pereira}.

For the phase flip example discussed in the previous section, the map $\Phi= \Phi_{\Delta_4}$ is the trace-preserving conditional expectation onto the $4\times 4$ diagonal subalgebra, $\Delta_4= \{ \sum_{i,j=0}^1 a_{ij} \kb{ij}{ij}: a_{ij}\in\C\}\cong \C\oplus\C\oplus\C\oplus\C$. Since $\Delta_4^{\prime}=\Delta_4$, the Kraus operators for $\Phi$ should be a basis for the diagonal matrices, and indeed they are.

From a geometric perspecitve, quasiorthogonal algebras arise through a natural broadening of the notion of orthogonality for algebras; namely, unital algebras $\mathcal{A}$, $\mathcal{B}$ satisfy quasiorthogonality if $tr(ab)=0$ whenever $tr(a)=0=tr(b)$ and $a\in\mathcal{A}$, $b\in\mathcal{B}$. The algebras are not orthogonal themselves as they both contain the identity operator, however the sets $\mathcal{A}\ominus \C I$ and $\mathcal{B}\ominus \C I$ are orthogonal in the Hilbert-Schmidt inner product $<\! a,b \!> = tr(b^* a)$. We state equivalent forms of this definition in our current notation.

\begin{defn} Two unital subalgebras $\mathcal{A}, \mathcal{B}\subseteq M_N(\C)$ are said to be \emph{quasiorthogonal} if any of the following equivalent conditions hold:
\begin{enumerate}
\item $tr\big((a-\frac{tr(a)}{N}I)(b-\frac{tr(b)}{N}I)\big)=0 \quad \forall a\in\mathcal{A}, \quad \forall b\in\mathcal{B}$
\item $\frac{1}{N}tr(ab)=\frac{1}{N}tr(a)\frac{1}{N}tr(b) \quad \forall a\in\mathcal{A}, \quad \forall b\in\mathcal{B}$
\item $\Phi_{\mathcal{A}}(b)=\frac{tr(b)}{N}I$ and $\Phi_{\mathcal{B}}(a)=\frac{tr(a)}{N}I \quad \forall a\in \mathcal{A}, \quad \forall b\in\mathcal{B}$
\item $\Phi_{\mathcal{A}}\circ \Phi_{\mathcal{B}}(\rho) = \Phi_{\mathcal{B}}\circ \Phi_{\mathcal{A}}(\rho)= \frac{tr(\rho)}{N}I$ for all $\rho$.
\end{enumerate}
\end{defn}

In particular, condition (3) says that the trace-preserving conditional expectation onto $\mathcal{A}$ privatizes an element of $\mathcal{B}$ (and vice-versa) precisely when the two subalgebras are quasiorthogonal to each other.

In the extremal case with $\mathcal{A} = M_N(\C)$ and $\mathcal{B} = \C I = \mathcal{A}^\prime$, $\Phi_\mathcal{B}$ is the complete depolarizing channel with Kraus operators given by a complete orthonormal set (in the Hilbert-Schmidt inner product) of operators, and $\Phi_\mathcal{B}(a) = \frac{tr(a)}{N}I$ for all $a\in \mathcal{A}$. As a simple subsystem example, consider the first qubit algebra $\mathcal{A}$ in two-qubit space generated by $\{II,XI,YI,ZI\}$. The commutant $\mathcal{A}^\prime = I_2\otimes M_2(\C)$ is generated by the orthonormal set $\{II,IX,IY,IZ\}$, which also act as Kraus operators (after normalizing) for the map $\Phi_\mathcal{A}$. Thus, $\mathcal{A}$ and $\mathcal{A}^\prime$ form a quasiorthogonal pair and in particular the second qubit is a private subsystem for the channel $\Phi_\mathcal{A}$. Not all quasiorthogonal algebra pairs arise through the commutant in this way though; indeed, as an example note that the phase flip example above is determined by the algebra pairing $\mathcal{A} = \Delta_4$ and $\mathcal{B}$ generated by $\{II, I X, YY, YZ\}$, which is unitarily equivalent to $I_2 \otimes M_2(\C)$.

This behaviour is not unique to qubit systems and can be generalized to multi-dimensional systems. The following example exhibits the behaviour of privatized subalgebras for elementary systems composed of qutrits, we leave its proof for Appendix~\ref{app:qutritProof}.

\begin{exa}
Let $X=\begin{pmatrix} 0 & 1 & 0 \\ 0 & 0 & 1 \\ 1 & 0 & 0\end{pmatrix}$ and $Z=\begin{pmatrix} 1 & 0 & 0 \\ 0 & \omega & 0\\ 0 & 0 & \omega^2\end{pmatrix}$, $\omega = e^{2i \pi/3}$, be the generalized Pauli operators for qutrits. Consider the channel on two qutrits given by
\begin{align*}
\Phi (\rho) = \dfrac{1}{9} \sum_{i,j=0}^2 (X^{2i} Z^i \otimes X^j Z^j) \rho (X^{2i} Z^i \otimes X^j Z^j)^\dagger.
\end{align*}
The Kraus operators for this channel are Abelian, and thus do not admit a private subspace. However, one may check that the subalgebra $\mathcal{A}$ generated by $X^2\otimes X, XZ^2\otimes Z$ is privatized by $\Phi$, and moreover that $\mathcal{A}$ is isomorphic as an algebra to the one-qutrit subalgebra generated by $I\otimes X$, $I\otimes Z$.
\label{ex:qutrit}
\end{exa}


As a next-step generalization of our motivating phase flip example, for the rest of the paper we consider conditional expectations onto algebras generated by Abelian subgroups of the $n$-qubit Pauli group. The subalgebras generated in this way will be Abelian, so by choosing a maximal Abelian subgroup of the $n$-qubit Pauli group that contains the initial subgroup, we obtain a maximal Abelian subalgebra generated by Pauli operators. The commutant of such an algebra will be itself, and so the Kraus operators for the trace-preserving conditional expectation will simply be the elements of the Abelian group, suitably weighted. Thus, having mutually commuting Kraus operators, there can be no private subspace for such a channel. However, if the Abelian subalgebra has a quasiorthogonal partner with a non-scalar component as in Eq.~(\ref{algunitequiv}), the channel will still have a private subsystem. This motivates us to learn more about the structure of maximal Abelian subgroups of the $n$-qubit Pauli group, $\mathcal{P}_n$, a topic to which we now turn.

\section{Commutation Relations of the Pauli Group}

Given $n\geq 1$, let $\mathcal{P}_n$ be the $n$-qubit Pauli group, which is the unitary subgroup of $M_{2^n}(\C)$ generated by all $n$-fold tensor products of the Pauli matrices $X,Y,Z$.

Note that the group-theoretic commutator  $[\sigma,\tau]=\sigma\tau\sigma^{-1}\tau^{-1}$ (as opposed to the more typically considered Lie algebra commutator) of any two $n$-qubit Pauli matrices is equal to $\pm I$. Also, $\left[\sigma,\tau\right] = \left[c_1\sigma,c_2\tau\right]$ for $\sigma,\tau \in \mathcal{P}_n$ and $c_1,c_2 \in \{\pm 1, \pm i\}$. Hence, the function $\chi: P_n\times P_n \rightarrow \C$ defined by $\chi([\sigma],[\tau])I=\left[\sigma,\tau\right]$ is well-defined on the central quotient $P_n = \mathcal{P}_n/\{\pm I, \pm iI \}$. Notice that $P_1$ is isomorphic to the Klein four-group $V = \{e, v_1,v_2,v_3\}$ where $e$ is the identity element, each $v_i$ is its own inverse, and $v_i v_j = v_k$. Similarly, $P_n \simeq V^n$, where $V^n$ is the (Abelian) direct product of $n$ copies of the Klein four-group with itself.

We recall that a {\it bicharacter} on a group $G$ is a function $B(\cdot,\cdot): G\times G \rightarrow \C$ satisfying the following:
\begin{enumerate}
\item $B(e,g)=B(g,e)= 1$ for all $g \in G$,
\item $B(g,hk)=B(g,h)B(g,k)$ and $B(hk,g)=B(h,g)B(k,g)$ for all $g,h,k \in G$.
\end{enumerate}
Moreover, a bicharacter is non-degenerate if for all non-identity $g\in G$ there exists some element $h \in G$ such that $B(g,h)\neq 1$.

Notice that a bicharacter, when restricted in either argument to a fixed $g\in G$ becomes a character on $G$.

\begin{prop} The function $\chi(\cdot,\cdot)$ defined above is a non-degenerate bicharacter on $P_n$.
\end{prop}
\begin{proof} Firstly, $\chi([I],[\sigma])= \left[ I,\sigma\right] = I$, hence $\chi([I],[\sigma])=1$ and similarly $\chi([\sigma],[I])= 1$.
To prove the second condition holds, note that $\sigma\rho = \chi([\sigma],[\rho])\rho\sigma$ and $\tau\rho = \chi([\tau],[\rho])\tau\sigma$. Hence
\begin{eqnarray*}
\sigma\tau\rho &=& \chi([\sigma\tau],[\rho])\rho(\sigma\tau)\\ &=& \chi([\tau],[\rho])\sigma\rho\tau\\ &=& \chi([\tau],[\rho])\chi([\sigma],[\rho])\rho\sigma\tau,
\end{eqnarray*}
and the other identity is similarly proved. The claim that $\chi(\cdot,\cdot)$ is non-degenerate is equivalent to the claim that for any non-identity $[\sigma]\in P_n$, there is some $[\tau]\in P_n$ that anti-commutes with $[\sigma]$, which is easily seen to be the case.
\end{proof}

Consider the character matrix for $P_1$ given by
\[
H = \begin{pmatrix} 1 & 1 & 1 & 1 \\ 1 & 1 & -1 & -1 \\ 1 & -1 & 1  & -1 \\ 1 & -1 & -1 & 1\end{pmatrix},
\]
where the columns and rows index the Pauli matrices and the entries of $H$ record whether elements of the Pauli basis commute or anti-commute. Then $H$ is the so-called bicharacter matrix of $\chi(\cdot,\cdot)$ on $P_1$; that is $H_{\sigma,\tau} = \chi([\sigma],[\tau])$. Moreover, since the restriction of $\chi([\sigma],[\tau])$ to any particular $\sigma$ yields a character of the Abelian group $P_1$, $H$ is a character matrix.

As noted above, the group $P_1$ is isomorphic to the Klein four-group, $V = \{e,x,y,z\}$, an Abelian group defined by the relations $x^2=y^2=z^2 = e$ and $xy=z$, $xz=y$, $yz=x$. Thus, $P_n$ is isomorphic to $V^n$, the direct product of $V$ with itself $n$ times, and $H_n := H^{\otimes n}$ is a character matrix for the Abelian group $P_n$, and records the commutation relations between basis elements for the $n$-qubit Pauli group. Lastly, we point out that a maximal Abelian subgroup of $P_n$ corresponds to a maximal submatrix of $H^{\otimes n}$ containing all $1$'s.

We recall another notion from group theory to continue. Let $G$ be a group with dual group $\widehat{G}$ and $K$ a subgroup of $G$. Then the {\it annihilator} of $K$ in $G$, $Ann_G(K)$ is the set of characters $\chi_i \in \hat{G}$ satisfying $\chi_i(k) =1$ for all $k\in K$. The following is a well-known fact about the annihilator subgroup \cite{rudin}.

\begin{lem}
Let $G,K,\widehat{G}$ be as above; then $Ann_{G}(K)$ is isomorphically homeomorphic to the dual group of $G/K$.
\end{lem}

In particular, this implies that when $G$ is finite and Abelian, that $|Ann_G(K)|=|G/K|=|G|/|K|$. The following is a direct consequence of this statement.

\begin{cor}\label{aacom}
Let $G=P_n$ and $K$ be any subgroup of $P_n$. Let $I$ be column indices of $H^{\otimes n}$ associated to the elements of $K$. Let $J=\chi^{-1}(Ann_G(K))$ be the row indices associated to the annihilator of $K$. Then we have $|I||J|=4^n$.
\end{cor}

\begin{cor}
Let $K$ be an Abelian subgroup of $P_n$ with $|K|=2^k < 2^n$. Then $K$ can be extended to an Abelian subgroup of size $2^n$.
\end{cor}

\begin{proof} We will show that so long as $k<n$, there exists a $g\in P_n\setminus K$ such that $g$ commutes with every element of $K$; and so that $<K,g>$, the group spanned by $K$ and $g$, must be Abelian.

We use the previous corollary as follows: let $I$ be the indices associated to the elements of $K$; since $K$ is Abelian the submatrix $H_n[I]$ is an all $1$ submatrix. Thus, by the previous result, there exist $4^n |I|^{-1}= 4^n 2^{-k}$ rows whose intersection with each column in $I$ contains only $1$'s. As $k<n$, we have $2^k < \frac{4^n}{2^k}$, and so there is at least one row not already in $I$ with this property. Call this row $i$. By the symmetry of $H_n$, the intersection of column $i$ with the rows of $I$ is all $1$'s, and $H_{n_{ii}=1}$, and so the submatrix indexed by $I\cup\{i\}$ is an all $1$ submatrix. Let $g$ be the element of $P_n$ associated with the $i^{th}$ column; then $g\in P_n$ is an element of $P_n\setminus K$ that commutes with every element of $K$. Hence $<K,g>$ is an Abelian group of size at least $2^{k+1}$.

We can iterate the procedure until $\frac{4^n}{2^k}=2^k$, or $k=n$.
\end{proof}

We are now ready to prove the following theorem which will be useful in constructing examples that generalize our motivating phase flip example.

\begin{thm}
Any maximal Abelian subalgebra $\mathcal{A}$ of $M_{2^n}(\C) $ generated by elements of $\mathcal{P}_n$ has dimension $2^n$.
\end{thm}

\begin{proof} To begin, we note that if $\sigma \in \mathcal{P}_n$ is an element of $\mathcal{A}$, then so are $-\sigma, \pm i \sigma$, and hence independent generators of $\mathcal{A}$ must all come from different conjugacy classes. Thus, we may regard our generators as coming from $P_n$, rather than $\mathcal{P}_n$.
It follows that a generating set for a maximal Abelian subalgebra is simply a maximal Abelian subgroup $K \leq P_n$, which by the previous corollary has size $2^n$. Since such a subgroup is already closed under matrix multiplication, it follows that $\dim <K> = 2^n$.
\end{proof}

\begin{cor}\label{cor:aaprime}
Any subalgebra $\mathcal{A}$ contained in the span of elements from $P_n$ has the property that $|\mathcal{A}||\mathcal{A}^{\prime}|=4^n$.
\end{cor}

Now we may turn our attention to the private structure of quantum channels whose normalized Kraus operators form an Abelian subgroup of the Pauli group.

\section{Privatizing Qubits with Abelian Subgroups}

We begin by discussing the simplest case, the channel whose Kraus operators are a normalized basis for the diagonal algebra, on arbitrary $n$-qubit space and what algebras can be privatized.

Let $\Delta_{2^n}$ be the diagonal algebra on $2^n\times 2^n$ complex matrices, with basis the group $\Delta$ generated by $\{Z_1,Z_2,\cdots,Z_n\}$ where $Z_1 = Z\otimes I \otimes \ldots \otimes I$, etc. Then $\Delta_{2^n}^{\prime}=\Delta_{2^n}$, and thus the conditional expectation $\Phi_{\Delta_{2^n}}$ onto $\Delta_{2^n}$,  has as its Kraus operators the elements of $\Delta$, multiplied by $\frac{1}{2^{n/2}}$.

Of course, we now know the answer is any algebra quasiorthogonal to the diagonal algebra $\Delta_{2^n}$. We can give an abstract description of such algebras, and also a concrete construction. The abstract result follows from the following result of \cite{pereira2} and provides an upper bound on the number of privatized qubits.

\begin{lem} Let $\mathcal{A}$ be a subalgebra of $M_N(\C)$ unitarily equivalent to $\bigoplus_{k=1}^m I_{i_k}\otimes M_{j_k}(\C)$. Then the following conditions are equivalent:
\begin{enumerate}
\item $\mathcal{A}$ is quasiorthogonal to $\Delta_N$;
\item $i_k \geq j_k$ for all $1\leq k \leq m$.
\end{enumerate}
\end{lem}

As $k$-qubits can be encoded into the unital algebra $I_{2^{n-k}}\otimes M_{2^k}(\C)$, the conditional expectation onto $\Delta_{2^n}$ can privatize $k$ qubits if and only if $n-k\geq k$. In other words, the conditional expectation onto $\Delta_{2^n}$ can privatize at most $\floor{\frac{n}{2}}$ qubits.

One explicit construction of such a subalgebra is to make use of the private subalgebra from our motivating example. For $1\leq i \leq \floor{ \frac{n}{2}}$ let $\widehat{X_i} = \bigotimes_{k=1}^n \sigma_k$, where
\[
\sigma_k = \left\{ \begin{array}{cl} X & \mbox{if}\, k = 2i \\
I & \mbox{otherwise} \end{array}\right.
\]
and let $\widehat{Y_i} = \bigotimes_{k=1}^n \sigma_k$, where
\[
\sigma_k = \left\{ \begin{array}{cl} Y & \mbox{if}\, k = 2i-1,2i \\
I & \mbox{otherwise} \end{array}\right.
\]
Then $\mathcal{B}=\Alg\{\widehat{X_i},\widehat{Y_j}: 1 \leq i,j \leq \floor{\frac{n}{2}}\}$, where $\Alg\{S\}$ is the unital algebra generated by operators in $S$, is quasiorthogonal to $\Delta_{2^n}$, precisely because $\widehat{X_i},\widehat{Y_i}$ and all their products are never diagonal, unless they are the identity. This algebra encodes $\floor{\frac{n}{2}}$ qubits in the obvious way, where $\widehat{X_i}$, and $\widehat{Y_i}$ act as the Pauli matrices $X_i$, $Y_i$ on $\floor{\frac{n}{2}}$ qubits with an $X$ or $Y$ in the $i^{th}$ tensor spot respectively.


We may use this explicit construction to prove the following in the general maximal Abelian case.

\begin{thm}\label{privsize} Let $\Phi$ be a a completely positive trace-preserving map on $M_{2^n}(\C)$ whose Kraus operators are equally weighted elements of a maximal Abelian subgroup $G\leq \mathcal{P}_n$. Then $\Phi$ can privatize $\floor{\frac{n}{2}}$ qubits.
\end{thm}

\begin{proof} Let $\mathcal{A} =\Alg\{G\}$ be the algebra generated by $G$. Then $\mathcal{A}$ is a maximal Abelian subalgebra of $M_N(\C)$. Thus $\mathcal{A}^{\prime}=\mathcal{A}$. The elements of $G$ are clearly a basis for $\mathcal{A}$, and hence $\Phi$ is the trace-preserving conditional expectation onto $\mathcal{A}^{\prime}=\mathcal{A}$. Moreover, since $\mathcal{A}$ is maximal Abelian, by simultaneously diagonalizing, there is a change of basis such that $U\mathcal{A}U^*=\Delta_{2^n}$.
Denote by $\mathcal{B}$ the subalgebra generated by $\{\widehat{X}_i,\widehat{Y}_j\}$ where $\widehat{X}_i$ and $\widehat{Y}_j$ are as above. Then $U^*\mathcal{B}U$ is privatized by $\Phi$; and following the discussion above, we see that $\mathcal{B}$ is sufficient to encode $\floor{\frac{n}{2}}$ qubits.
\end{proof}

We now examine the case that $\mathcal{A}$ is a non-maximal Abelian algebra generated by Pauli operators. Corollary \ref{cor:aaprime} says that for any subalgebra $\mathcal{A}$ generated by elements of $\mathcal{P}_n$, $\dim(\mathcal{A})\dim(\mathcal{A}^{\prime})=4^n = \dim(\mathcal{H})^2$. The algebras for which this equality holds can be characterized as follows from \cite{pereira2}.

\begin{lem}\label{homogbalalg} Let $\mathcal{A}$ be a unital subalgebra of $M_n(\C)$, then the following are equivalent:
\begin{enumerate}
\item $\dim(\mathcal{A})\dim(\mathcal{A^{\prime}})=N^2$
\item $\mathcal{A} = U(\bigoplus_{i=1}^m I_{sk_i}\otimes M_{rk_i})U^*$ for some unitary matrix $U$
\end{enumerate}
\end{lem}

\begin{lem}\label{abelalgebras} Every commutative subalgebra $\mathcal{A}\subseteq M_{N}(\C)$ generated by Pauli operators is unitarily equivalent to an algebra of the form $I_{2^{n-k}}\otimes \Alg\{Z_i \}_{i=1}^k$ for some $k\leq n$.
\end{lem}

\begin{proof}
By Lemma~\ref{homogbalalg}, any commutative subalgebra $\mathcal{A}$ generated by commuting Pauli operators satisfies $rk_i =1$ for all $i$, and hence for some $s$,  decomposes as
\[
\mathcal{A} = U\big(\bigoplus_{i=1}^m I_s \otimes \C\big)U^* = U(I_{2^{n-k}}\otimes \Delta_{2^{k}})U^*,
\]
for some $k$. Finally, as $\{Z_i\}_{i=1}^k$ generate $\Delta_{2^k}$, we have that $I_{2^{n-k}}\otimes\Alg\{ Z_i\}_{i=1}^k = I_{2^{n-k}}\otimes\Delta_{2^k}$.
\end{proof}

\begin{thm}
Let $\Phi:M_N(\C)\rightarrow M_N(\C)$ be a quantum channel whose Kraus operators $\{K_i\}_{i=1}^{2^k}$ are equally weighted elements of an Abelian subgroup of $\mathcal{P}_n$. Then there is a $\frac{k}{2}$-qubit algebra privatized by $\Phi$.
\end{thm}

\begin{proof}
By Lemma \ref{abelalgebras}, the algebra generated by the Kraus operators must be unitarily equivalent to $I_{2^{n-k}}\otimes \Delta_{2^k}$; hence after diagonalizing and restricting to the subalgebra spanned by only the last $k$ qubits, we obtain the maximal Abelian subalgebra on $k$ qubits; Theorem \ref{privsize} tells us how to privatize $\frac{k}{2}$ qubits in this scheme. Tensoring any private algebra encoding $\frac{k}{2}$ qubits for the $2^k\times 2^k$ maximal Abelian case with a $2^{n-k}\times 2^{n-k}$ identity yields an algebra encoding $\frac{k}{2}$ qubits that is private for $\Phi$.
\end{proof}

\begin{rmk}
In a sense the results in this Section can be thought of as the subsystem analog of the results from \cite{ambainis}, where it was shown that $k$~unitaries could be used to privatize a~$\frac{k}{2}$-qubit subspace. Here, we have shown that we can privatize a $\frac{k}{2}$-qubit subsystem algebra using $k$~elements from an Abelian subgroup of~$\mathcal{P}_n$.
\end{rmk}

\section{Conclusion}

In this work, we demonstrate the underlying mathematical principles behind private quantum subsystems. Namely, by using the theory of quasiorthogonal algebras and developing the set of tools therein, we can make definite statements on the dimension of the subsystems that can be privatized by commuting Kraus operators. We show that by taking elements of a maximal Abelian subgroup~$G$ of the $n$-qubit Pauli group~$\mathcal{P}_n$, the elements of such a group can privatize $\floor{\frac{n}{2}}$~qubits when used as equally weighted Kraus operators of the channel, and we give explicit constructions for such private subsystems. Moreover, we show that this result can be generalized to fewer qubits when the size of the Abelian subgroup is not taken to be maximal.

\vspace{0.1in}

{\noindent}{\it Acknowledgements.} J.L. was supported by an AIMS-University of Guelph Postdoctoral Fellowship. T.J.-O. was supported by the Vanier--Banting Secretariat and NSERC through the Vanier Canada Graduate Scholarship. D.W.K was supported by NSERC and a University Research Chair at Guelph. R.L. was supported by NSERC, CIFAR, and the Canadian and Ontario governments. R.P. was supported by NSERC.

\bibliographystyle{plain}
\bibliography{paulirefs}

\newpage 
\appendix
\section{Proof of Privatized Qutrit Subsystem}
\label{app:qutritProof}

In this Appendix, we prove the result stated in Example~\ref{ex:qutrit}. Namely, for the generalized Pauli $X$~and~$Z$ operators on qutrits, the channel
\begin{align*}
\Phi (\rho) = \dfrac{1}{9} \sum_{i,j=0}^2 (X^{2i} Z^i \otimes X^j Z^j) \rho (X^{2i} Z^i \otimes X^j Z^j)^\dagger
\end{align*}
has no private subspace but can privatize the subalgebra~$\mathcal{A}$ generated by~$X^2 \otimes X,\ XZ^2 \otimes Z$, which is isomorphic to the one-qutrit subalgebra~$I \otimes X, \ I \otimes Z$.

\begin{proof}
Explicitly, $X^3=Z^3=I$, and $XZ=\omega ZX$, so the group generated by $X,Z$ has a bicharacter given by $\left[g,h\right] = \chi(g,h)I$, yielding the following character matrix:
\begin{equation}
F = \bordermatrix{  & I& X      & X^2  & Z      & XZ      & X^2Z   & Z^2    & XZ^2   & X^2Z^2 \cr
              I & 1& 1      & 1    &      1 & 1       &  1     &  1     &  1     &  1     \cr
              X & 1& 1      & 1    & \omega &\omega   &\omega  &\omega^2&\omega^2&\omega^2\cr
             X^2& 1& 1      & 1    &\omega^2&\omega^2 &\omega^2&\omega  &\omega  &\omega  \cr
               Z& 1&\omega^2&\omega&1       &\omega^2 &\omega  &  1     &\omega^2&\omega  \cr
              XZ& 1&\omega^2&\omega&\omega  &  1      &\omega^2&\omega^2&\omega  &  1     \cr
            X^2Z& 1&\omega^2&\omega&\omega^2&\omega   &  1     &\omega  &  1     &\omega^2\cr
             Z^2& 1&\omega  &\omega^2&1     &\omega   &\omega^2&  1     &\omega  &\omega^2\cr
            XZ^2& 1&\omega  &\omega^2&\omega&\omega^2 &  1     &\omega^2&  1     &\omega  \cr
          X^2Z^2& 1&\omega  &\omega^2&\omega^2& 1     &\omega  &\omega  &\omega^2&  1        }
\end{equation}

 In two qutrit space, the matrix of commutation relations is $F\otimes F$; any all-$1$s submatrix of $F\otimes F$ yields an Abelian subalgebra. Our choice corresponds to the tensor product of the $\{1,6,8\}$ submatrix of $F$ with itself union the tensor product of the $\{1,5,9\}$ submatrix of $F$ with itself. 

To find a quasiorthogonal algebra to this Abelian algebra, we take a quotient by the subalgebra to obtain the following equivalence classes:
\begin{equation*} \begin{array}{c|c|c|c|c|c|c|c}
I\otimes I        & I\otimes X        & I \otimes X^2   &X \otimes I         &X^2 \otimes I          \\ \hline
I\otimes XZ       & I\otimes X^2Z     & I\otimes Z      &X\otimes XZ         &X^2\otimes XZ           \\
I\otimes X^2Z^2   & I\otimes Z^2      &I\otimes XZ^2    &X\otimes X^2Z^2     &X^2\otimes X^2Z^2  \\
X^2Z\otimes I     & X^2Z\otimes X     &X^2Z\otimes X^2  &Z\otimes I          &XZ\otimes I           \\
X^2Z\otimes XZ    & X^2Z\otimes X^2Z  &X^2Z\otimes Z    &Z\otimes XZ         &XZ\otimes XZ       \\
X^2Z\otimes X^2Z^2& X^2Z\otimes Z^2   &X^2Z\otimes XZ^2 &Z\otimes X^2Z^2     &XZ\otimes X^2Z^2   \\
XZ^2\otimes I     &XZ^2\otimes X      &XZ^2\otimes X^2  &X^2Z^2\otimes I     &Z^2\otimes I       \\
XZ^2\otimes XZ    &XZ^2\otimes X^2Z   &XZ^2\otimes Z    &X^2Z^2\otimes XZ    &Z^2\otimes XZ       \\
XZ^2\otimes X^2Z^2&XZ^2\otimes Z^2    &XZ^2\otimes XZ^2 &X^2Z^2\otimes X^2Z^2&Z^2\otimes X^2Z^2  \end{array}
\end{equation*}
\begin{equation*}
\begin{array}{c|c|c|c}  X\otimes X & X\otimes X^2        & X^2\otimes X    & X^2\otimes X^2 \\ \hline
                      X\otimes X^2Z & X\otimes Z          & X^2\otimes X^2Z & X^2\otimes Z \\
                      X\otimes Z^2 & X\otimes XZ^2       &X^2\otimes Z^2   & X^2\otimes XZ^2\\
                      Z\otimes X & Z\otimes X^2        &XZ\otimes X      &XZ\otimes X^2 \\
                      Z\otimes X^2Z & Z\otimes Z          &XZ\otimes X^2Z   &XZ\otimes Z \\
                      Z\otimes Z^2 & Z \otimes XZ^2      &XZ\otimes Z^2    &XZ\otimes XZ^2\\
                      X^2Z^2\otimes X & X^2Z^2\otimes X^2   &Z^2\otimes X     &Z^2\otimes X^2\\
                      X^2Z^2\otimes X^2Z & X^2Z^2\otimes Z     &Z^2\otimes X^2Z  &Z^2\otimes Z\\
                      X^2Z^2\otimes Z^2 & X^2Z^2\otimes XZ^2  &Z^2\otimes Z^2   &Z^2\otimes XZ^2\end{array}\end{equation*}
Notice that cosets $3$, $5$, $8$, and $9$ are obtained by squaring cosets $2$, $4$, $7$ and $6$ respectively, so a choice of two generators for a quasiorthogonal subalgebra should make sure to avoid these pairs. Our choice corresponds to choosing $X^2\otimes X$ from coset $8$ and $XZ^2\otimes Z$ from coset $3$. To see that this choice of generators is isomorphic to a qutrit algebra, note that the block unitary
\begin{equation*} U = \left(\begin{array}{c|c|c} I & X^2Z^2 & XZ \\ \hline
XZ^2 & Z & X^2 \\ \hline
X^2Z & X & Z \end{array}\right)\end{equation*}
acts by on $I\otimes X$ and $I\otimes Z$ by
\begin{align*}
U(I\otimes X)U^* &= \omega^2 XZ^2\otimes Z\\
U(I\otimes Z)U^* &= X^2 \otimes X
\end{align*}
\end{proof}

\end{document}